\documentclass{article}

\usepackage[margin=1in]{geometry}
\usepackage{amsmath,amssymb,amsthm,enumerate,url,siunitx,multirow,hhline,graphicx,cancel}
\usepackage{unicode-math}

\newtheorem{theorem}{Theorem}

\newtheorem{lemma}[theorem]{Lemma}

\newtheorem{remark}[theorem]{Remark}
\newtheorem{example}[theorem]{Example}

\newcommand{\paren}[1]{\left(#1\right)}
\newcommand{\bracket}[1]{\left[#1\right]}
\renewcommand{\brace}[1]{\left\{#1\right\}}

\newcommand{\F}{{\mathbb{F}}}
\newcommand{\nnz}[1]{\mathrm{nnz}\paren{#1}}
\newcommand{\nz}[1]{\mathrm{nz}\paren{#1}}
\newcommand{\cat}{\mathrm{cat}}
\newcommand{\lex}{\mathrm{lex}}
\newcommand{\lexmin}{\mathrm{lexmin}}
\newcommand{\can}[1]{\mathrm{can}\paren{#1}}
\newcommand{\GL}{\mathrm{GL}}
\newcommand{\ptwise}{\mathrm{ptwise}}
\newcommand{\rk}[1]{\mathrm{rk}\paren{#1}}
\newcommand{\M}[1]{\begin{bmatrix}#1\end{bmatrix}}
\newcommand{\MA}[2]{\bracket{\begin{array}{#1}#2\end{array}}}

\renewcommand{\span}[1]{\mathrm{span}\paren{#1}}
\renewcommand{\dim}[1]{\mathrm{dim}\paren{#1}}

\newcommand{\cpd}[1]{\left[\!\left[#1\right]\!\right]}
\newcommand{\rowspace}[1]{\mathrm{rowspace}\paren{#1}}
\newcommand{\maxrank}{\mathcal{R}}

\let\oldtextbf=\textbf
\renewcommand\textbf[1]{{\boldmath\oldtextbf{#1}}}

\title{Maximal rank of $4\times 4\times 4$ and $k\times 4\times 3$ tensors over $\F_2$}
\author{Jason Yang}
\date{September 2026}

\begin{document}

\maketitle

\begin{abstract}
We determine the maximum rank of $4\times 4\times 4$ and $k\times 4\times 3$ multidimensional arrays over the finite field $\F_2$. These results are made possible by a new method of enumerating non-isomorphic tensors.
\end{abstract}

\section{Introduction}
The rank of a tensor (multidimensional array) is the smallest number of vector outer products that it can be expressed as a sum of.
Computing the rank of an arbitrary tensor is NP-hard \cite{shitov}, and looser properties of tensor rank remain difficult to analyze, such as the maximum possible rank for a given shape. Over finite fields, research is especially sparse, as one cannot use tools from algebraic geometry as easily as over the complex numbers.

Our goal is solve maximum rank for new tensor shapes over the finite field $\F_2$, focusing on 3-dimensional shapes.
Existing work has solved the shapes $m\times n\times 2$ \cite{mxnx2}; $3\times 3\times 3$ \cite{bremner}; and $4\times 3\times 3$, $5\times 3\times 3$, and $4\times 4\times 3$: \cite{yang}.
This work solves $4\times 4\times 4$ and all remaining $k\times 4\times 3$.

We first describe a recursive method of enumerating non-isomorphic tensors that is exponentially faster than the standard orbit-finding algorithm used in \cite{bremner}.
We then compute the rank of each tensor individually, using an improved version of our rank-finding algorithm in \cite{yang}. Surprisingly, we found this method to be faster than a breadth-first search over isomorphism classes.

\subsection{Definitions}
Let $T$ be a tensor with shape $n_0\times\dots\times n_{D-1}$.
A decomposition of $T$ with rank $R$ is a list of vector tuples $(a^{(r)}_d)_{0\le d<D},\ 0\le r<R$ such that $T=\sum_r a^{(r)}_0\times\dots\times a^{(r)}_{D-1}$, where $\times$ denotes the \textit{outer product}.
Alternatively, it is a list of \textit{factor matrices} $A_d\in\F^{n_d\times R}$ such that $T_{i_0,\dots,i_{D-1}}=\sum_r \prod_d (A_d)_{i_d,r}$ for all indices $(i_0,\dots,i_{D-1})$.
The rank of $T$, denoted $\rk{T}$, is the minimum possible value of $R$.

The \textit{axis-$d$} rank of $T$ is the dimension of the linear span of its axis-$d$ slices. We call $T$ \textit{axis-$d$ concise} if its axis-$d$ rank equals $n_d$, and \textit{concise} if $T$ is axis-$d$ concise for all axes $d$.

\subsubsection{Notation}
\begin{itemize}
    \item The outer product of two tensors $A, B$ is denoted as $(A\times B)_{i_0,\dots,i_{D-1},\ j_0,\dots,j_{E-1}} = A_{i_0,\dots,i_{D-1}} B_{j_0,\dots,j_{E-1}}$.

    \item The matrix product of a tensor $T$ with a \textit{matrix} $M$ along axis $d$ is denoted as \[M\times_d T=\bracket{\sum_{i_d} M_{i',i_d} T_{i_0,\dots,i_d,\dots,i_{D-1}}}_{i_0,\dots,i',\dots,i_D{-1}}.\]

    We write $M$ to the left of $T$, unlike most sources, as it looks more similar to standard matrix multiplication.
    For brevity, we omit parentheses when applying multiple matrix products to one tensor, i.e. $A \times_0 B\times_1 C\times_2 T$ means $A\times_0 (B\times_1 (C\times_2 T))$.

    If $M$ is a row vector, it is treated like a $1\times n_d$ matrix and the $d$-th axis of the resulting product is discarded.

    \item A tensor decomposition with factor matrices $A_0,\dots,A_{D-1}$ is also written as $\cpd{A_0,\dots,A_{D-1}}$.

    \item A 1-hot vector with a 1 at index $i$ and 0s everywhere else is denoted as $e_i$; the length is implied by context.

    \item We use 0-indexing and NumPy notation to denote elements or slices of a tensor:
    \begin{itemize}
        \item writing $a:b$ on an axis denotes the range of indices $\brace{a,\dots,b-1}$ on that axis;

        when omitted, $a$ and $b$ are set to 0 and $n$, respectively, where $n$ is the length of the tensor along the axis;

        \item writing $\cdots$ after an axis denotes the ranges of all indices on all remaining axes.
    \end{itemize}

    For example, $A_{:,j}$ is the $j$-th column of a 2D tensor $A$,
    and $T_{i,\dots}$ is the $i$-th slice along axis 0 of an tensor $T$ with one or more dimensions.

    \item $O^*$ asymptotic complexity ignores polynomial factors.

    \item $[n]$ denotes the set $\brace{0,\dots,n-1}$.
\end{itemize}

\section{Tensor enumeration}
Two tensors $T,T'$ of shape $n_0\times\dots\times n_{D-1}$ are \textit{isomorphic} if there exist invertible matrices $Q_0,\dots,Q_{D-1}$ such that $Q_0\times_0 (Q_1\times_1 \dots T)=T'$. We denote this relation as $T\sim T'$.
It is clear that isomorphic tensors have equal rank.

Define the \textit{canonicalization} of a tensor $T$ to be the lexically minimum tensor it is isomorphic to, under row-major ordering of elements.
The following are immediate:

\begin{lemma}
\label{prefix}
If $T=\can{T}$, then for any $p$, $T_{:p,\dots}=\can{T_{:p,\dots}}$.
\end{lemma}
\begin{proof}
Suppose $\can{T_{:p,\dots}}=Q_0\times_0 Q_1\times_1 \dots T_{:p,\dots} <_\lex T_{:p,\dots}$: then $\M{Q_0\\&I_{n_0-p}}\times_0 Q_1\times_1 \dots T <_\lex T$: contradiction.
\end{proof}

\begin{lemma}
\label{zero-prefix}
Let $r=\dim{\span{\brace{T_{i,\dots}:0\le i<n_0}}}$: then $\can{T}_{:n_0-r,\dots}=0$.
\end{lemma}
\begin{proof}
By definition of $r$, the maximum number of all-zeros axis-0 slices any tensor $T'\sim T$ can have is $n_0-r$.
\end{proof}

For $D=2$, we can solve canonicalization immediately. As a consequence, the canonicalization of a matrix only depends on its rank, and lower rank corresponds to lower lexical order.
\begin{lemma}
\label{2d}
For any matrix $M\in\F^{m\times n}$,
$\can{M}=\M{0_{(m-r)\times (n-r)}&0_{(m-r)\times r}\\0_{r\times (n-r)}&H}$,
where $r=\rk{M}$
and $H=\underbrace{\M{&&1\\&⋰\\1}}_r$. 
\end{lemma}
\begin{proof}
From Lemma \ref{zero-prefix}, $\can{M}_{:m-r,:}=0$.
Then for each $0\le i<r$,
$\can{M}_{m-r+i,:}$ is the lexically smallest row vector not in the row-span of $\can{M}_{m-r:m-r+i-1,:}$.
The desired expression for $\can{M}$ is then obtained via a standard proof by induction.
\end{proof}

\subsection{Enumeration}
For the remainder of this section, we focus on $D=3$.
To enumerate canonical $n_0\times n_1\times n_2$ tensors,
the previous lemmas suggest the following ideas:
\begin{itemize}
    \item generate $k\times n_1\times n_2$ canonicals that are axis-0 concise, for $k\le n_0$, then prefix-pad with all-zeros slices;

    \item to generate $n_0\times n_1\times n_2$ canonicals,
    iterate over $\cat_0(T,M)$ for canonical $T\in\F^{(n_0-1)\times n_1\times n_2}$ and arbitrary $M\in\F^{n_1\times n_2}$, where $\cat_0$ denotes concatenation along axis 0 and $M$ is implicitly expanded to $1\times n_1\times n_2$;

    we call $\cat_0(T,M)$ an \textit{augment} of $T$;

    \item use Lemma \ref{2d} to enumerate $1\times n_1\times n_2$ canonicals.
\end{itemize}

Below is the first version of our algorithm; for brevity, we only return the list of $n_0\times n_1\times n_2$ axis-0 concise canonicals.
We describe later how to check isomorphism.

\begin{itemize}
    \item $L\gets []$
    \item for each $S\in\F^{(n_0-1)\times n_1\times n_2}$ that is canonical and axis-0 concise, in lex order:
    \begin{itemize}
        \item for each $M\in\F^{n_1\times n_2}$ in lex order:
        \begin{itemize}
            \item $T\gets \cat_0(S,M)$
            \item if $T$ is not axis-0 concise:
            \begin{itemize}
                \item continue
            \end{itemize}
            \item if $\forall T'\in L: T\not\sim T'$:
            \begin{itemize}
                \item add $T$ to $L$
            \end{itemize}
        \end{itemize}
    \end{itemize}
    \item return $L$
\end{itemize}

We iterate over $\le \phi(n_0-1,n_1,n_2) |\F|^{n_1 n_2}$ many $T$ and run $\Theta(\phi(n_0,n_1,n_2))$ many isomorphism checks for each $T$ on average, where $\phi(n_0,n_1,n_2)$ is the number of $n_0\times n_1\times n_2$ canonicals.

The above algorithm does not take advantage of the fact that many canonical tensors share prefixes, so running many individual isomorphism checks may lead to wasted work.
To solve this problem, we introduce a new relation between tensors $T\in\F^{n_0\times n_1\times n_2}$ and $S\in\F^{n'\times n_1\times n_2}$ for $n'\le n_0$:

\[(T\ge_0 S)
:=
(\exists P\in\F^{n'\times n_0} \textrm{ full row-rank},\ Q_1\in\GL(n_1,\F),\ Q_2\in\GL(n_2,\F) \textrm{ s.t. } P\times_0 Q_1\times_1 Q_2\times_2 T=S).\]

It is clear that $T\sim T'$ implies $T\ge_0 T'_{:p,\dots}$ for all $p$.
Taking the contrapositive allows us to modify the algorithm as follows:
\begin{itemize}
    \item for $0\le k<n_0$:
    \begin{itemize}
        \item $L_k\gets [S\in\F^{k\times n_1\times n_2}:\ S \textrm{ canonical and axis-0 concise, in lex order}]$
    \end{itemize}
    \item $L_{n_0}\gets []$
    \item for $(S\in L_{n_0-1},\ M\in \F^{n_1\times n_2})$ in lex order:
    \begin{itemize}
        \item $T\gets\cat_0(S,M)$
        \item if $T$ is not axis-0 concise:
        \begin{itemize}
            \item continue
        \end{itemize}
        \item if $\forall 1\le k\le n_0,\ \forall S'\in L_k \textrm{ s.t. } S'<_\lex T_{:k,\dots} \ \& \ S'_{:k-1,\dots}=T_{:k-1,\dots}:\ T\not\ge_0 S'$:
        \begin{itemize}
            \item add $T$ to $L_{n_0}$
        \end{itemize}
    \end{itemize}
    \item return $L_{n_0}$
\end{itemize}

Essentially, we determine if $T$ is canonical by checking whether it can be $\ge_0$-matched to all canonical prefixes lexically smaller than $T$, while skipping redundant prefixes.
Note that we update $L_{n_0}$ while it is also being read.

In the innermost if-statement, we iterate over $\le n_0 |\F|^{n_1 n_2}$ many $S'$, since for each $k$ all but the last axis-0 slice of $S'$ is fixed. Thus, we run that many $\ge_0$-checks for each $T$ instead of $\Theta(\phi(n_0,n_1,n_2))$ many isomorphism checks.

\subsubsection{Tensor matching}
To efficiently determine whether $T\ge_0 S$ for tensors $T\in\F^{n_0\times n_1\times n_2},\ S\in\F^{n'\times n_1\times n_2}$, we first describe a depth-first search (DFS) over the axis-0 matrix that fixes one row at a time and handles the other axes as a base case:
\begin{itemize}
    \item $\mathrm{iso12}(T,T'):=(\exists Q_1\in\GL(n_1,\F),\ Q_2\in\GL(n_2,\F) \textrm{ s.t. } Q_1\times_1 Q_2\times_2 T=T')$
    \item def $\mathrm{dfs}(P\in\F^{p\times n_0})$:
    \begin{itemize}
        \item // precondition: $P$ has full row-rank, $\mathrm{iso12}(P\times_0 T, S_{:p,\dots})$ is True
        \item if $p=n'$:
        \begin{itemize}
            \item return True
        \end{itemize}
        \item for $v\in \F^{n_0}$:
        \begin{itemize}
            \item $P'\gets \MA{c}{P\\\hline v}$
            \item if $P'$ has full row-rank and $\mathrm{iso12}(P'\times_0 T, S_{:p+1,\dots})$ and $\mathrm{dfs}(P')$:
            \begin{itemize}
                \item return True
            \end{itemize}
        \end{itemize}
        \item return False
    \end{itemize}
    \item return $\mathrm{dfs}(0_{0\times n_0})$
\end{itemize}

Each $\ge_0$ check makes $\le \F^{n_0 n'}$ many iso12 calls in the worst case, but may make significantly fewer in practice due to pruning.

Instead of trying to make iso12 faster, we notice that across all checks $T\ge_0 S$ during tensor enumeration, $S$ has relatively few distinct values; in fact, $S$ is in $L_0\sqcup\dots\sqcup L_{n_0-1}$, so it has exactly $\phi(n_0-1,n_1,n_2)$ possibilities.

We can take advantage of this property if we apply axis-1 and axis-2 operations during the DFS to match prefixes of $T$ to $S$:
\begin{itemize}
    \item def $\mathrm{dfs}(P\in\F^{p\times n_0},\ T\in\F^{n_0\times n_1\times n_2})$:
    \begin{itemize}
        \item // precondition: $P\times_0 T=S_{:p,\dots}$
        \item if $p=n'$: return True
        \item for $v\in \F^{n_0}$:
        \begin{itemize}
            \item $P'\gets \MA{c}{P\\\hline v}$
            \item if $P'$ has full row-rank and $\mathrm{iso12}(P'\times_0 T, S_{:p+1,\dots})$:
            \begin{itemize}
                \item find invertible $Q_1,Q_2$ s.t. $Q_1\times_1 Q_2\times_2 (P'\times_0 T)=S_{:p+1,\dots}$
                \item if $\mathrm{dfs}(P',\ Q_1\times_1 Q_2\times_2 T)$: return True
            \end{itemize}
        \end{itemize}
        \item return False
    \end{itemize}
    \item return $\mathrm{dfs}(0_{0\times n_0},\ T)$
\end{itemize}

Now, every call $\mathrm{iso12}(U,U')$ for tensors $U,U'\in\F^{n'\times n_1\times n_2}$ satisfies the properties $U_{:n'-1,\dots}=U'_{:n'-1,\dots}$ and $U'\in L_0\sqcup\dots\sqcup L_{n_0-1}$.

Define $V:=U_{:n'-1,\dots}, M=U_{n'-1,\dots}, M'=U'_{n'-1,\dots}$:
then $\mathrm{iso12}(U,U')$ is equivalent to $M$ and $M'$ belonging to the same orbit of $\F^{n_1\times n_2}$ under the group action

\[G_V:=\brace{(M\mapsto P\times_0 Q\times_1 M): P\in\GL(n_1,\F),\ Q\in\GL(n_2,\F),\ P\times_1 Q\times_2 V=V}.\]

In this situation, it is feasible to precompute and save all group orbits.
It is also straightforward to return a group action sending an arbitrary $M$ to some other $M'$, by precomputing actions that send every possible $M$ to a consistent orbit representative.
Thus, each $\mathrm{iso12}(U,U')$ check can be evaluated in $O^*(1)$ time.

\subsubsection{Combined matching}
If we use our improved algorithm for $T\ge_0 S$-checks to determine whether a given tensor $T$ is canonical, there is still wasted work because many of the $S$-tensors share prefixes with each other. Specifically, two different checks $T\ge_0 S$ and $T\ge_0 S'$ will visit the exact same tree of $(P,T)$-states until after some depth.
Our fix is to run a unified search over all such $S$, ultimately removing the $\ge_0$-relation.

\begin{itemize}
    \item // goal: determine if $T_{n_0\times n_1\times n_2}$ is canonical
    \item // precondition: $T$ is axis-0 concise
    \item for $0\le k<n_0$:
    \begin{itemize}
        \item $\mathcal{S}_k\gets \mathrm{sorted}([S_{(k+1)\times n_1\times n_2}: S \textrm{ canonical, axis-0 concise, } S_{:k,\dots}=T_{:k,\dots},\ S<_\lex T_{:k+1,\dots}])$
    \end{itemize}
    \item $\mathcal{E}\gets [(0_{0\times n_0}, T)]$  (``search states")
    \item for $k=0,\dots,n_0-1$:
    \begin{itemize}
        \item // precondition: $\forall (P,T')\in \mathcal{E}:\ P\times_0 T'=T_{:k,\dots}$
        \item $\mathcal{E}'\gets []$
        \item for $(P,T')\in\mathcal{E},\ v\in\F^{n_0}$:
        \begin{itemize}
            \item $P'\gets \MA{c}{P\\\hline v}$
            \item if $P'$ not full row-rank:
            \begin{itemize}
                \item continue
            \end{itemize}
            \item if $\exists S\in \mathcal{S}_k \ | \ \mathrm{iso12}(P'\times_0 T',\ S)$:
            \begin{itemize}
                \item return False
            \end{itemize}
            \item if $\mathrm{iso12}(P'\times_0 T',\ T_{:k+1,\dots})$:
            \begin{itemize}
                \item find invertible $Q_1,Q_2$ such that $Q_1\times_1 Q_2\times_2 (P'\times_0 T')=T_{:k+1,\dots}$
                \item add $(P',\ Q_1\times_1 Q_2\times_2 T')$ to $\mathcal{E}'$
            \end{itemize}
        \end{itemize}
        \item $\mathcal{E}\gets \mathcal{E}'$
        \item // at this point, $\mathcal{E}$ must contain $(I_{:k+1,:},\ T_{:k+1,:})$ and therefore be nonempty
    \end{itemize}
    \item return True
\end{itemize}

Notice that $\mathcal{S}_{n_0-1}$ is a subset of $L_{n_0}$, which is updated during the main enumeration procedure. The other lists $\mathcal{S}_k$ can be reused across different $T$.

Another way to understand this procedure is that for each $k=0,\dots,n_0-1$, we attempt to $\ge_0$-transform $T$ into a $k\times n_1\times n_2$ canonical prefix that is nonstrictly lexically smaller than $T$, while tracking all possible axis-0 operations $P$ that could be used in such a transformation.

\subsubsection{Axis-0 canonicalization}
As a final speedup, we check if each $T$ is canonical \textit{along axis 0 only} before checking if it is truly canonical. Define

\[\mathrm{can}_0(T):=\lexmin_{Q\in\GL(n_0,\F)} Q\times_0 T.\]

Computing $\mathrm{can}_0(T)$ can be done in $O^*(|\F|^{n_0})$ time, by iterating over every row vector $v\in \F^{n_0}$ in lexical order of $v\times_0 T$ and greedily finding the next vector $v$ that is linearly independent from previously selected vectors.

\section{Tensor rank}
\label{sec:cpd}
We use our depth-first search algorithm from \cite{yang} to determine the rank of each canonical tensor, but with stronger pruning. For completeness, we describe the algorithm in full.

Let $T\in\F^{n_0\times\dots\times n_{D-1}},\ n_0\ge\dots\ge n_{D-1}$ be a concise tensor and $R\ge n_0$ be the rank threshold.
WLOG $\rk{T}\le R$ iff there exist $A_d\in\F^{n_d\times n_0},\ B_d\in\F^{n_d\times (R-n_0)}$ s.t. $T=\cpd{\dots,\MA{c|c}{A_d&B_d},\dots_d}$ and $A_0$ is invertible.
Letting $Q:=(A_0)^{-1}$ and $X=QB_0$, we have $Q\times_0 T=\cpd{I,A_1,A_2,\dots}+\cpd{X,B_1,B_2,\dots}$.
Subtracting both sides by $\cpd{X,B_1,\dots}$ and splitting the equation by its axis-0 slices yields $Q_{i,:}\times_0 T - \cpd{X_{i,:},B_1,B_2,\dots}=(A_1)_{:,i}\times(A_2)_{:,i}\times\dots$ for all $0\le i<n_0$.

Consider fixing all $(B_d)_{d\ge 1}$: then the system has a solution iff the set
\[S:=\brace{(q,x): q\in\F^{n_0},\ x\in\F^{R-n_0},\ a_d\in\F^{n_d},\ q\times_0 T - \cpd{x,B_1,B_2,\dots}=a_1\times\dots\times a_{D-1}}\] satisfies $\dim{\span{\brace{q: (q,x)\in S}}}=n_0$.
This can be determined with two methods:
\begin{enumerate}[(1)]
    \item Enumerate all $(q,x)$ such that $\rk{q\times_0 T - \cpd{x,B_1,B_2,\dots}}\le 1$;
    This method takes $O^*(|\F|^R)$ time,
    because checking whether a tensor $T$ has rank $\le 1$ is equivalent to checking whether its axis-$d$ rank is $\le 1$ for all $d$.

    \item For all $(a_2,\dots,a_d)$, construct a basis for the linear system $q\times_0 T - \cpd{x,B_1,B_2,\dots}=a_1\times\dots\times a_{D-1}$ over $(q,x,a_1)$;
    then find the span of the union of these bases.
    This method takes $O^*(|\F|^{\sum_{d\ge 2} n_d})$ time.
\end{enumerate}

Repeating for all $(B_d)_{d\ge 1}$, the total running time is $O^*(|\F|^{(R-n_0)\paren{\sum_{d\ge 1} n_d} + \min\paren{R,\ \sum_{d\ge 2} n_d}})$.
In our implementation, we only use (1) as it is simpler and the difference in running time is insignificant for the tensor shapes we are interested in.

\subsection{Depth-first search pruning}
To convert the aforementioned algorithm into DFS, we partition the variables in $(B_d)_{d\ge 1}$ into lists of columns $((B_d)_{:,r})_{d\ge 1}$ for some $r$, then fix them for each $r=0,\dots,\ R-n_0-1$ at a time. Doing so allows us to potentially prune nodes that fail a sufficiency check of our choosing.

Suppose a node has fixed the leftmost $p$ columns of each $(B_d)_{d\ge 1}$; note that $(B_0)_{:,:p}$ is still unknown.
Isolating the unknown variables as much as possible in the equation $T=\cpd{\MA{c|c}{A_0&B_0},\dots}$ yields $T-\cpd{(B_0)_{:,:p},\dots}=\cpd{C_0,\dots}$ for $C_d:=\MA{c|c}{A_d&(B_d)_{:,p:}}$.

Consider multiplying the left hand side of the equation along axis 0 with a vector $v\in \F^{n_0}$, then calculating the rank. Observe that

\[\rk{v\times_0 T-\cpd{v(B_0)_{:,:p},(B_1)_{:,:p},\dots}}\le\nnz{vC_0},\]

where $\nnz{\cdot}$ counts the number of nonzero elements.

We then ``forget" $(B_0)_{:,:p}$ by replacing $v(B_0)_{:,:p}$ with a free vector $w$. Because we cannot control how $w$ is related to $v$, we must take the minimum of the left hand side to make the inequality as loose as possible:

\[\min_{w\in\F^p}\rk{v\times_0 T-\cpd{w,(B_1)_{:,:p},\dots}}\le\nnz{vC_0}.\]

This condition motivates us to define $f_{T;\ P_1,P_2,\dots}(v):=\min_{w\in\F^p}\rk{v\times_0 T-\cpd{w,P_1,P_2,\dots}}$ and $g_{C_0}(v):=\nnz{vC_0}$; then our pruner is

\[\exists \textrm{full row-rank } C_0\in\F^{n_0\times (R-p)} \textrm{ s.t. } f_{T;\ (B_1)_{:,:p},\dots}\le_\ptwise g_{C_0}.\]

We can restrict $C_0$ to have full row-rank because $A_0$ is invertible.

For each $p$, we precompute $g_{C_0}$ for all $C_0$, up to permutation of columns, and store the maps in a data structure that efficiently supports $\le_\ptwise$-membership queries, treating each $g_{C_0}$ as a list. The number of $C_0$ we enumerate is $O(|\F|^{n_0(R-p)} / (R-p)!)$.

We use one of two data structures:
\begin{enumerate}[1.]
    \item If memory permits, create a bit-array $A$ with a 1 at each $g_{C_0}$ and 0 elsewhere, then compute the multidimensional summed-area table (using boolean ORs of suffixes instead of sums of prefixes);

    \item Else, store all $g_{C_0}$ in a prefix tree, and process each $\le_\ptwise$-membership query by DFS-ing on the tree with pruning.
\end{enumerate}

To reduce memory and computation for the bit-array, we compress the $f$ and $g$ maps:
\begin{itemize}
    \item omit the $v=\vec{0}$ input;
    \item replace $f$ and $g$ with $f':=\max(f-1,0)$ and $g':=\min(g-1,M-1)$, where $M:=n_2\cdot \dots\cdot n_{D-1}$;

    Doing so does not weaken the original point-wise inequality, because we are guaranteed $g_{C_0}(v)\ge 1$ for $v\ne \vec{0}$ (since $C_0$ has full row-rank) and $f_{T;\ (B_1)_{:,:p},\dots}(v)\le M$ for all $v$ (since $f(v)$ is the rank of a $n_1\times n_2\times \dots\times n_{D-1}$-shaped tensor).
\end{itemize}

This improves the bit-array memory from $O((R+1)^{|\F|^{n_0}})$ to $O(\min(R,M)^{|\F|^{n_0}-1})$, e.g. for the tensor shape $4\times 4\times 4$ and rank threshold $R=9$ over the field $\F_2$, it improves memory from $O(10^{16})$ to $O(4^{15})$.

When both methods are too costly, we fall back to the following looser methods, each presented with a short proof under the condition that an admissible $C_0$ exists. For brevity, denote $R'=R-p$. The first two are from our previous work \cite{yang} and the rest are new:

\begin{itemize}
    \item \textbf{(``rref") $\dim{\span{\brace{v:f_{T;\ P_1,P_2,\dots}(v)\le R'-n_0+1}}}=n_0$.}

    For $Q:=(A_0)^{-1}$, we have $QC_0=\M{I&\dots}$; then for each $i$, setting $v=Q_{i,:}$ satisfies $g_{C_0}(v)\le R'-n_0+1$.

    \item \textbf{(``Laskowski", equivalent to \cite{laskowski}) $\sum_{v\in\F^{n_0}} (R'-f_{T;\ P_1,P_2,\dots}(v)) \ge R'|\F|^{n_0-1}$.}

    The left hand side is $\ge \sum_v (R'-g_{C_0}(v))
    =\sum_{v;\ i\in [R']} [v(C_0)_{:,i}=0]
    =\sum_i (\# v \ | \ v(C_0)_{:,i}=0)
    \ge \sum_i |\F|^{n_0-1}
    =R'|\F|^{n_0-1}$.

    \item \textbf{If $\F=\F_2$ and $R'\ge n_0+2$: $\dim{\span{\brace{v:f_{T;\ P_1,P_2,\dots}(v)\le R'-n_0}}}\ge n_0-1$.}

    Let $S:=\span{\brace{v:\nz{vC_0}\ge n_0}}$, where $\nz{}$ denotes the number of zeros; then
    it suffices to prove $\dim{S}\ge n_0-1$.

    WLOG assume $R'=n_0+2$ and $C_0=\M{I&X}$.
    Let $H=\brace{i_0,i_1,\dots}$ denote all $i$ such that $X_{i,:}=\M{1&1}$;
    then $e_{i_0}+e_{i_j}\in S$ for all $j\ne 0$, and $e_i\in S$ for all $i\not\in H$, which amounts to $n_0-1$ many vectors (unless $|H|=0$, in which case there are $n_0$ many vectors).

    \item \textbf{For any $1\le k<n_0$: $\sum_{v\in\F^{n_0}} \binom{R'-f_{T;\ P_1,P_2,\dots}(v)}{k} \ge \binom{R'}{k}|\F|^{n_0-k}$.}

    The left hand side is $\ge \sum_v \binom{R'-g_{C_0}(v)}{k}
    =\sum_{v;\ I\in \binom{[R']}{k}} [v(C_0)_{:,I}=0] 
    =\sum_I (\# v \ | \ v(C_0)_{:,I}=0)
    \ge \sum_I |\F|^{n_0-k}
    =\binom{R'}{k}|\F|^{n_0-k}$.
\end{itemize}

For our computation, we used full pruning (bit-array/prefix tree) for $4\times 4\times 4$ at all rank thresholds we tested, and a mix of all pruning methods for $*\times 4\times 3$.

\section{Maximum rank}
Before finding maximum tensor rank for our desired shapes with the method we have outlined, we establish some bounds to reduce computational work.

Let $\maxrank_\F(k,m,n)$ denote the maximum rank of a $k\times m\times n$ tensor over a ground field $\F$; we write $\maxrank(k,m,n)$ if $\F$ is clear from context or when a statement is true over all ground fields.

We recall the method in \cite{yang} to construct high-rank tensors, which is itself similar to \cite{atkinson}. We give all proofs here for completeness.

\begin{lemma}[Substitution method]
\label{substitution-method}
For a tensor $T\in\F^{n_0\times\dots}$ and nonzero row vector $v$ s.t. $v\times_0 T\ne 0$, there exists some column vector $w$ s.t. $vw\ne 0$ and $\rk{(I-w(vw)^{-1}v)\times_0 T}\le \rk{T}-1$.
\end{lemma}
\begin{proof}
Let $T=\sum_r a_r\times \dots$ be an optimal decomposition.
Then there must exist some $r^*$ such that $a_{r^*}$ satisfies $va_{r^*}\ne 0$. Setting $w=a_{r^*}$ results in $(I-w(vw)^{-1}v)\times_0 T$ annihilating the $r=r^*$ term.
\end{proof}

\begin{remark}
An equivalent way to state the substitution method is that there exists some $M\in\F^{(n_0-1)\times n_0}$ with full row-rank s.t. $v\not\in \rowspace{M}$ and $\rk{M\times_0 T}\le \rk{T}-1$.
\end{remark}

The same result applies along any axis. We call this action ``subbing out" $v$.
When $v$ is 1-hot at index $i$, we notate $v$ as a tensor index slice, e.g. $(i,\dots)$ if along axis 0.
Each sub increases the rank lower bound by 1.

\begin{example}
We prove $T=\M{\M{&1\\1}&\M{1&\phantom{0}\\&}}$ has rank 3 over any field.
The upper bound is obvious.
For the lower bound, sub out $(1,:,:)$ to get $\M{\M{x&1\\1}}$ for some arbitrary $x$ (i.e. set $v=\M{0&1},\ M=\M{1&x}$, and the axis to 0).
By inspection, it is clear this tensor has rank 2 no matter what $x$ is. Alternatively, we can sub out two more times at $(:,1,:)$ and $(:,:,1)$.
\end{example}

\begin{lemma}[similar to Proof of Thm. 2 in \cite{atkinson}]
\label{lower-bound-expand}
For any $m,n,m',n',k,p$ with $m'\ge m,\ n'\ge n$ and $p\le m'n'-mn$,
we have $\maxrank(k+p, m', n') \ge \maxrank(k,m,n) + p$.
\end{lemma}
\begin{proof}
Let $T\in\F^{k\times m\times n}$ have maximum rank. Construct $T'$ by padding $T$ along axes 1 and 2, then adding $p$ many matrices that are 1-hot at distinct $(i,j)$ for each $(i,j)\in [m']\times [n'] \setminus ([m]\times [n])$.
Then $\rk{T'}\ge \rk{T}+p$ because we can sub out $(h,:,\dots)$ for each $h=k+p-1,\dots,k$, then truncate along axes 1 and 2 to get $T$.
\end{proof}

\begin{lemma}[special case of Thm. 2 of \cite{atkinson}]
\label{max-rank-one-below}
$\maxrank(mn-1,m,n)=mn-1$ for $m,n\ge 1$.
\end{lemma}
\begin{proof}
$\ge mn-1$: clear.

$\le mn-1$: for any $T\in\F^{(mn-1)\times m\times n}$, there exists some nonzero $M\in\F^{m\times n}$ s.t. $\span{\brace{T_{i,\dots}}_i}\subseteq S:=\brace{A\in\F^{m\times n}: \sum_{i,j} A_{i,j} M_{i,j}=0}$. WLOG assume $M=\M{I_r & O \\ O & O}$ for some $r\ge 1$.
Observe that the set $\brace{e_i\times e_j: (i,j)\in [m]\times [n] \setminus \brace{(i,i): 0\le i<r}} \cup \brace{(e_0+e_i)\times (e_0-e_i): 1\le i<r}$ consists of $mn-1$ many rank-1 matrices whose span is precisely $S$, as $(e_0+e_i)\times (e_0-e_i)+e_0\times e_1-e_i\times e_0=e_0\times e_0-e_i\times e_i$.
\end{proof}

Finally, it can be confirmed with our algorithm in Section \ref{sec:cpd} that the $3\times 3\times 2$ tensor \[\M{\M{\\&1\\1} & \M{&1\\1\\&} & \M{1\\\\1&1}}\] has rank 5 over $\F_2$. Combining this with previous lemmas yields the following:

\begin{lemma}
\label{lower-bound-minus-3}
For all $m\ge 2,\ n\ge 2$, and $(m,n)\ne (2,2)$, $\maxrank_{\F_2}(mn-3,m,n)=mn-1$.
\end{lemma}

As a consequence, $\maxrank(k,4,3)=11$ for $0\le k\le 11$, so we only need to run our algorithm on $k\times 4\times 3$ for $k\le 8$.
Furthermore, $\maxrank(k,3,3)=8$ for all $k\ge 6$, which, combined with prior work, solves max rank for all $k\times 3\times 3$ shapes.

\section{Results}
Our code is available at

\begin{center}
\url{https://github.com/coolcomputery/tensor-max-rank-mod2}.
\par
\end{center}

Tables \ref{tab:canonical} and \ref{tab:rank-stats} summarize our results for max rank via computational search.
Because the set of canonical tensors of shape $n_0\times n_1\times n_2$ naturally includes those of shape $(<n_0)\times n_1\times n_2$ up to zero-padding, we list elapsed times only for $n_0$.

To save space, each tensor $T$ is represented as the contraction $v\times_0 T$ where $v$ consists of formal variables.

\begin{table}
    \centering
    \begin{tabular}{|c|c|c|}
        \hline
        shape & time to enumerate canonicals (min) & time to compute ranks (min) \\
        \hline
        $4\times 4\times 4$ & 80 & 1900 \\
        \hline
        $8\times 4\times 3$ & 120 & 13 \\
        \hline
    \end{tabular}
    \caption{Elapsed time to enumerate and to compute ranks of all canonical tensors for each shape.}
    \label{tab:canonical}
\end{table}

\begin{table}
    \centering
    \begin{tabular}{|c|c|c|c|}
        \hline
        shape & \# canonicals & max rank & example \\
        \hline
        $2\times 4\times 4$ & 58 & 6 & $\M{&&&v_1\\&&v_1\\&v_1&&v_0\\v_1&&v_0}$ \\
        \hline
        $3\times 4\times 4$ & 5626 & 8 & $\M{&&v_2&v_1\\&&v_1+v_2&v_2\\v_2&v_1&&v_0\\v_1+v_2&v_2&v_0}$ \\
        \hline
        $4\times 4\times 4$ & 2295780 & 9 & $\M{&&v_3&v_2\\&&v_2+v_3&v_3\\v_3&v_2&v_1&v_0\\v_2+v_3&v_3&v_0+v_1&v_1}$ \\
        \hline
        \hline
        $2\times 4\times 3$ & 28 & 5 & $\M{&&v_1\\&v_1\\&&v_0\\v_1&v_0}$ \\
        \hline
        $3\times 4\times 3$ & 355 & 6 & $\M{&&v_2\\&v_2\\&&v_1\\v_2&v_1&v_0}$ \\
        \hline
        $4\times 4\times 3$ & 5626 & 8 & same as $3\times 4\times 4$ \\
        \hline
        $5\times 4\times 3$ & 42691 & 8 & - \\
        \hline
        $6\times 4\times 3$ & 115735 & 9 & $\M{&v_5&v_4\\v_5&&v_3\\v_4&v_3&v_5\\v_2&v_1&v_0}$ \\
        \hline
        $7\times 4\times 3$ & 152800 & 9 & - \\
        \hline
        $8\times 4\times 3$ & 158071 & 10 & $\M{&v_7&v_6 \\ v_7&v_6&v_5 \\ v_6&v_4&v_3 \\ v_2&v_1&v_0}$ \\
        \hline
    \end{tabular}
    \caption{Number of canonical tensors, maximum rank, and example maximum-rank tensors for each nontrivial shape.}
    \label{tab:rank-stats}
\end{table}

\begin{remark}
The max-rank $4\times 4\times 4$ tensor is equal to the Kronecker product of $\M{\M{\phantom{0}&\\&1}&\M{&1\\1}}$ and $\M{\M{&1\\1}&\M{1\\1&1}}$, which are precisely all of the non-isomorphic tensors of shape $2\times 2\times 2$ and rank 3 over $\F_2$.
\end{remark}

\begin{remark}
The max-rank $6\times 4\times 3$ and $8\times 4\times 3$ tensors can be constructed by augmenting the max-rank $3\times 3\times 3$ and $5\times 3\times 3$ tensors, respectively, to $*\times 4\times 3$ via Lemma \ref{lower-bound-expand}.
\end{remark}

For each tensor whose rank is less than its number of nonzero elements, we describe an optimal decomposition.

\begin{itemize}
    \item $3\times 4\times 4$:
    two copies of $\M{v_2&v_1\\v_1+v_2&v_2}
    =\M{v_2&v_2\\v_2&v_2}+\M{\phantom{0}&v_1+v_2\\&}+\M{&\\v_1&\phantom{0}}$,
    plus the 1-hot decomposition of $\M{&v_0\\v_0}$

    \item $4\times 4\times 4$:
    three copies of $\M{v_2&v_1\\v_1+v_2&v_2}$ (up to relabeling of variables)

    \item $6\times 4\times 3$:
    $\M{&v_5&v_4\\v_5&&v_3\\v_4&v_3&v_5}
    =\M{&\\\phantom{0}&\phantom{0}&v_3\\&}
    +\M{v_4&v_4&v_4\\&\\v_4&v_4&v_4}
    +\M{v_4+v_5&v_4+v_5&\phantom{0}\\&\\&}
    +\M{v_5&\phantom{0}&\phantom{0}\\v_5\\&}
    +\M{&\\&\\\phantom{0}&v_3+v_4&\phantom{0}}
    +\M{&\\&\\\phantom{0}&\phantom{0}&v_4+v_5}$, plus the 1-hot decomposition of $\M{v_2&v_1&v_0}$

    \item $8\times 4\times 3$: $\M{&v_7&v_6\\v_7&v_6}=\M{\phantom{0}&v_6+v_7&\phantom{0}\\&v_6+v_7}+\M{&\\v_7&v_7&\phantom{0}}+\M{\phantom{0}&v_6&v_6\\&}$, plus the 1-hot decomposition of $\M{&&v_5\\v_6&v_4&v_3\\v_2&v_1&v_0}$
\end{itemize}

\subsection{Lower bounds for $k\times 4\times 4$}
Combining our results for $\maxrank(k,4,3)$ with Lemmas \ref{lower-bound-expand} and \ref{lower-bound-minus-3} yields the following lower bounds for $\maxrank(k,4,4)$:
\begin{itemize}
    \item $\maxrank(5,4,4)\ge \maxrank(4,4,4)=9$
    \item $\maxrank(k,4,4)\ge \maxrank(4,4,3)+(k-4)=k+4$ for $6\le k\le 8$
    \item $\maxrank(k,4,4)\ge \maxrank(k-4,4,3)+4$ for $9\le k\le 12$
    \item $\maxrank(k,4,4)=15 \ \forall 13\le k\le 15$
\end{itemize}

Table \ref{tab:all-max-rank} lists the best known lower bounds for $\maxrank(k,4,4)$ to our knowledge, as well as max ranks for the other shapes we have analyzed.

\begin{table}
    \centering
    \begingroup
    \setlength{\tabcolsep}{5pt}
    \begin{tabular}{|c|ccccccccccccccccc|}
        \hline
        $k$ & 0 & 1 & 2 & 3 & 4 & 5 & 6 & 7 & 8 & 9 & 10 & 11 & 12 & 13 & 14 & 15 & 16 \\
        \hline
        $\maxrank(k,3,3)$ & 0 & 3 & 5 & 6 & 6 & 7 & 8 & 8 & 8 & 9 &&&&&&& \\
        $\maxrank(k,4,3)$ & 0 & 3 & 5 & 6 & 8 & 8 & \textbf{9} & \textbf{9} & \textbf{10} & 11 & 11 & 11 & 12 &&&& \\
        $\maxrank(k,4,4)$ & 0 & 4 & 6 & 8 & \textbf{9} & $\ge 9$ & $\ge 10$ & $\ge 11$ & $\ge 12$ & $\ge 12$ & $\ge 13$ & $\ge 13$ & $\ge 14$ & 15 & 15 & 15 & 16 \\
        \hline
    \end{tabular}
    \endgroup
    \caption{Table of known maximum ranks over $\F_2$ for various tensor shapes. Numbers in bold are new results.}
    \label{tab:all-max-rank}
\end{table}

\end{document}